\documentclass[11pt]{article}
\usepackage{amsfonts}
\usepackage{amsmath,amsthm,mathrsfs}
\usepackage{hyperref}

\usepackage[all]{xy}
\usepackage{graphicx}

\usepackage{amssymb}
\usepackage{latexsym}
\usepackage{slashed}

\DeclareMathAlphabet\EuFrak{U}{euf}{m}{n}	
\SetMathAlphabet\EuFrak{bold}{U}{euf}{b}{n}	

\newcommand{\lb}{{\left\langle \right.}}
\newcommand{\rb}{{\left. \right\rangle}}

\newcommand{\ovl}{\overline}

\newcommand{\bC} {{\mathbb C}}
\newcommand{\bF} {{\mathbb F}}

\newcommand{\bI} {{\mathbb I}}
\newcommand{\bR} {{\mathbb R}}

\newcommand{\bU} {{\mathbb U}}

\newcommand{\bGL}{{\mathbb{GL}}}
\newcommand{\bZ} {{\mathbb Z}}

\newcommand{\bN} {{\mathbb N}}

\newcommand{\supp}{{\mathrm{supp}}}

\newcommand{\mB}{\mathcal B}

\newcommand{\mF}{\mathcal F}

\newcommand{\mH}{\mathcal H}

\newcommand{\mP}{\mathcal P}

\newcommand{\mR}{\mathcal R}
\newcommand{\mS}{\mathcal S}

\newcommand{\mU}{\mathcal U}

\newcommand{\efg}{\EuFrak g}

\newcommand{\efu}{\EuFrak u}

\newcommand{\xP}{\mP}
\newcommand{\ad}{{\boldsymbol{ad}}}

\newtheorem{thm}{Theorem}[section]

\newtheorem{lem}[thm]{Lemma}
\newtheorem{prop}[thm]{Proposition}

\theoremstyle{definition}

\theoremstyle{remark}

\numberwithin{equation}{section}

\begin{document}

\author{{\sf Ezio Vasselli}
\\{\small{Dipartimento di Matematica,}}
\\{\small{Universit\`a di Roma ``Tor Vergata'',}}
\\{\small{Via della Ricerca Scientifica 1, I-00133 Roma, Italy.}}
\\{\sf ezio.vasselli@gmail.com}
}

\title{Background potentials and superselection sectors}
\maketitle

\begin{abstract}
The basic aspects of the Aharonov-Bohm effect can be summarized by the remark that wavefunctions become 
sections of a line bundle with a flat connection (that is, a "flat potential"). 
Passing at the level of quantum field theory in curved spacetimes, 
we study the Dirac field interacting with a classical (background) flat potential and show that it can be interpreted 
as a topological sector of the observable net of free Dirac field. 
On the converse, starting from a topological sector we reconstruct a classical flat potential, 
interpreted as an interaction of a Dirac field.
This leads to a description of Aharonov-Bohm-type effects in terms of localized observables.

\medskip

\noindent Keywords: Aharonov-Bohm effect, superselection sector, Dirac field, curved spacetime.
\end{abstract}

\section{Introduction}
\label{sec.intro}

Despite the Aharonov-Bohm effect is a sixty-year old result
and appears in undergraduate textbooks as a consolidate topic, it still attracts a considerable attention.
The well-known scenario is that of a double-slit experiment, enriched by a solenoid directed towards the vertical axis and 
carrying a magnetic field $\vec{B}$ (\emph{magnetic} Aharonov-Bohm effect).
The solenoid is shielded, so that:
(1) the charged quantum particles passing through the slits are confined in a region $N \subset \bR^3$ where the magnetic field $\vec{B}$ vanishes;
(2) $N$ is multiply connected (it has non-trivial fundamental group $\pi_1(N) \simeq \bZ$).
The experimental evidence is that if $\gamma_1 , \gamma_2$ are paths in $N$ ("trajectories"), then the shift phase
\begin{equation}
\label{eq.1}
\exp \oint_\ell \vec{A}  \ \ \ , \ \ \ \ell := \ovl{\gamma}_2 * \gamma_1 \, ,
\end{equation}
appears in coherent superpositions of wavefunctions with total support homotopic to the loop $\ell$.
Here $\vec{A} : N \times \bR \to \bR^3$ is the (possibly time-dependent) vector potential and $d\vec{A}=\vec{B}=0$;
thus the point is that states are affected by a quantity, $\vec{A}$,
which classically does not have a well-defined physical meaning,
whilst the physical quantity, $\vec{B}$, vanishes.
An elegant way to express this scenario is to regard wavefunctions as sections 
\begin{equation}
\label{eq.2}
f : N \to L_A \, ,
\end{equation}
where $L_A \to N$ is the flat line bundle with monodromy (\ref{eq.1}) \cite{BM};
therefore we can regard $f$ as a function $f|_o : o \to \bC$ only on simply connected regions $o \subset N$.

As argued in Feynman's lecture notes \cite[\S 15-5]{Fey}, the Aharonov-Bohm effect may be regarded as a manifestation
of the fact that the Schroedinger equation contains $\vec{A}$ as an interaction term.
Thus, passing to a relativistic setting, a natural analogue is to consider a spacetime $M$ with $\pi_1(M)$ non-trivial,
and to study a (quantum) Dirac field $\psi_A$ interacting with a background (classical, electromagnetic) flat potential $A$;
by \emph{flat}, we mean that the curvature $F$ vanishes. 
We prove that $\psi_A$ can be easily constructed, Theorem \ref{thm.B1},
at the price of twisting the Dirac bundle: the twist is the relativistic counterpart of (\ref{eq.2}).

The purpose of the present paper is to describe $\psi_A$ in terms of observable quantities
following the ideas of algebraic quantum field theory \cite{Haa}: namely, that the physics of a quantum system is determined
by the observables organized in a net of local algebras, and that the observables are local and obey to Einstein's causality principle.
We give simplified versions of results that will appear in full generality in a forthcoming paper
by C. Dappiaggi, G. Ruzzi and the author; thus we focus on the basic ideas rather than insist on technical details.


The following sections are organized as follows.

In \S \ref{sec.A} we illustrate a result that allows to reconstruct a gauge potential
starting from a generalized Aharonov-Bohm phase
$\sigma : \pi_1(M) \to G$,
where $G$ is a compact Lie group \cite{Bar91}. The physical interpretation is that starting from $\sigma$,
which yields the effective observable, one is able to reconstruct a potential whose path-ordered integral is the phase itself.

In \S \ref{sec.B} we construct the field $\psi_A$, defined on a globally hyperbolic spacetime $M$.
We show that $\psi_A$ is equivalent to a family of fields $\psi_o$ defined on arcwise and simply connected regions $o \subset M$
(diamonds), all gauge-equivalent to a free Dirac field $\psi$ (Theorem \ref{thm.B1}).
Thus one cannot discriminate $\psi_A$ from $\psi$ in simply connected regions at the level of observables.

In \S \ref{sec.C} we show that any field of the type $\psi_A$ can be interpreted as a topological superselection sector
of the observable net defined by the free field $\psi$, Theorem \ref{thm.C1}.
This yields the expected physical interpretation of the superselection sectors introduced years ago by Brunetti and Ruzzi \cite{BR08}.
The basic idea behind this result is that background flat potentials or, to be precise, their Aharonov-Bohm phases,
yield the parallel transport over loops of charges in the observable net defined by $\psi$.

In \S \ref{sec.D} we give a sketch of the construction involving non-Abelian Aharonov-Bohm phases,
which typically appear in the case of a Dirac field carrying a symmetry by a non-Abelian global gauge group.

Our conclusions are given in the final \S \ref{sec.E}:
we discuss the physical interpretation of characterizing a background flat potential as a superselection sector,
and briefly sketch future developments ($\pi_1(M)$ non-Abelian and non-flat background potentials).


\section{Aharonov-Bohm effect: geometric aspects}
\label{sec.A}

Geometric properties of the spacetime manifold manifest themselves in the Aharonov-Bohm effect by means of 
the classical gauge potential interacting with the quantum particle.
Thus it is convenient to recall some constructions in differential geometry, commonly used in classical gauge theories,
that play a role in the present paper.

We start with some conventions.
For any vector bundle $B \to M$, we denote the space of sections by $\mS(B)$,
the space of compactly supported sections by $\mS_c(B)$,
and the space of sections supported in $o \subset M$ by $\mS_o(B)$;
in the sequel, sections shall be assumed to be smooth.
For example, $k$--forms are sections $A \in \mS(\wedge^k T^*M)$, where $T^*M$ is the cotangent bundle,
and in particular for \emph{closed} $k$--forms ($dA=0$) we write $A \in Z_{dR}^k(M)$.
Finally, when $M$ has a spin structure, spinors are sections $f \in \mS(DM)$ of the Dirac bundle $DM$.
%

%

Let $G$ be a Lie group with Lie algebra $\efg$ and $\ad : G \to \bGL(\efg)$ denote the adjoint action.
Given a morphism $\sigma : \pi_1(M) \to G$, one can construct a principal $G$-bundle $P_\sigma \to M$ which is \emph{flat}
in the sense that it admits a set of locally constant transition maps \cite[\S I.2]{Kob}.
In turn, $P_\sigma$ defines the $\efg$-bundle
$\efg_\sigma \to M$, $\efg_\sigma := P_\sigma \times_\ad \efg$
{\footnote{Here $P_\sigma \times_\ad \efg$ is the quotient of $P_\sigma \times \efg$ by the equivalence relation
           $(y,v) \sim (yg,\ad_g(v))$, $\forall g \in G$.}}.
By \cite[\S 2.4]{Bar91}, there is a flat connection on $P_\sigma$ with connection form 
$A \in \mS(T^*M \otimes \efg_\sigma)$,
such that
\begin{equation}
\label{eq.A4}
\sigma(\ell) \ = \ \xP\exp \oint_\ell A \, \in G 
\ \ \ , \ \ \ 
\forall [\ell] \in \pi_1(M) \, ,
\end{equation}
where $\xP\exp \oint$ is the path-ordered integral \cite[\S II.2]{BM}.
We call $A$ the \emph{flat gauge potential} defined by $\sigma$.
Using $A$ we can construct transition maps for $P_\sigma$, in the following way:
we consider a reference point $\bar{x} \in M$,
a base $\{ o \subset M \}$ of arcwise and simply connected regions, 
and pick $x_o \in o$ for any $o$;
then for any pair $y,y' \in M$ we fix a curve
$\gamma(y'y) : y \to y'$
and define the loops
\begin{equation}
\label{eq.A2a}
\ell(oa;x) \ := \ \gamma(\bar{x}x_o) * \gamma(x_o x) * \gamma(x x_a) * \gamma(x_a \bar{x})
\ \ \ , \ \ \ 
x \in o \cap a \neq \emptyset \, ; 
\end{equation}
finally we set
\begin{equation}
\label{eq.A2}
g_{oa}(x) \, := \, 
\sigma(\ell(oa;x)) \, = \, 
\xP\exp \oint_{\ell(oa;x)} A \, \in G \ \ \ , \ \ \ x \in o \cap a \, ,
\end{equation}
%
%
and, since $\sigma$ preserves loop compositions, we have the 1--cocycle relations
\begin{equation}
\label{eq.A3}
g_{oa} g_{ae} \ = \ g_{ae}  \ \ \ , \ \ \ \forall o \cap a \cap e \neq \emptyset \, .
\end{equation}
Since the homotopy class of $\ell(oa;x)$ is constant for $x$ varying in a connected component of $o \cap a$,
each $g_{oa}$ is locally constant, and in particular $g_{oa}$ is constant when $a \subseteq o$. 
For $e \subseteq a \subseteq o$, (\ref{eq.A3}) holds in terms of constant maps,
and this shall be useful in the context of superselection sectors \S \ref{sec.C}.

\begin{lem}
\label{lem.A1}
With the above notation, $g = \{ g_{oa} \}$ is a set of transition maps for $P_\sigma$, 
and $\ad g$ is a set of transition maps for $\efg_\sigma$.
\end{lem}

\begin{proof}
Let $\pi : \hat{M} \to M$ denote the universal covering. 
It is well-known that $\hat{M}$ can be regarded as a $\pi_1(M)$-principal bundle
carrying an action by homeomorphisms $r_\ell$, $[\ell] \in \pi_1(M)$,
which preserve the fibres $\hat{M}_x := \pi^{-1}(\{ x \})$, $x \in M$.
By \cite[\S 14.3]{Ste}, $\hat{M}$ has transition maps
$\hat{g}_{oa}(x,y) :=$ $( x \, , \, r_{\ell(oa;x)}(y) )$
defined for $(x,y) \in (o \cap a) \times F$, where $F := \hat{M}_{\bar{x}}$ has been chosen as a standard fibre.
By definition, $P_\sigma$ is the $G$-bundle with associated principal $\pi_1(M)$-bundle $\hat{M}$,
that is, the quotient $\hat{M} \times_\sigma G$ of $\hat{M} \times G$ by the equivalence relation
$(r_\ell(y),h) \sim (y, h \sigma(\ell))$, $[\ell] \in \pi_1(M)$,
\cite[\S I.2]{Kob}.
Thus $P_\sigma$ has transition maps induced by $\hat{g}_{oa}$,
%
%
given by
$g_{oa}(x,h) :=$ $( x \, , \, h\sigma(\ell(oa;x)) )$, $(x,h) \in (o \cap a) \times G$,
as desired.
The statement for $\efg_\sigma$ is a trivial consequence of the first part of the proof.
\end{proof}

In the case $G = \bU(1)$ we have $\efg \simeq \bR$ and $\ad$ is trivial, so that $\efg_\sigma \simeq M \times \bR$.
Thus any morphism $\sigma : \pi_1(M) \to \bU(1)$ defines a 1--form $A \in \mS(T^*M)$ such that
\begin{equation}
\label{eq.A4ab}
\sigma(\ell) \ = \  \exp \oint_\ell A \, \in \bU(1) \ \ \ , \ \ \ \forall [\ell] \in \pi_1(M)
\end{equation}
(in the Abelian case, $\xP\exp \oint = \exp \oint$).
$A$ is now the vector potential, and (\ref{eq.A4ab}) represents the shift phase
carried by wavefunctions in the Aharonov-Bohm effect.
As remarked in \S \ref{sec.intro}, we describe this fact by regarding wavefunctions as sections of the line bundle
$L_A \to M$
defined by the transition maps (\ref{eq.A2}).
%
%

Now, in general it is not true that there is an isomorphism $L_A \simeq M \times \bC$ 
(when this happens, we say that $L_A$ is \emph{trivial}).
Indeed this turns out to be the case when, in particular,
$dA = 0$ 
{\footnote{
 In general, if $A$ is a flat gauge potential then $dA$ is a 2--form
 representing the real Euler class of $P_\sigma$,
 and $(2\pi)^{-1} \int_S dA$ $=$ $(2\pi)^{-1} \int_{\partial S} A \in \bZ$ for any 2--simplex $S$ \cite[Ex.1.5]{CS85}.
 Thus the previous integrals do not contribute to (\ref{eq.A4ab}), 
 as in the more familiar case $dA=0$.
 }: 
%
%
in fact, in this hypothesis there are local primitives 
$\phi_o \in C^\infty(o)$, $d\phi_o = A_o$,
defined on any arcwise and simply connected $o \subset M$, 
so that each
\begin{equation}
\label{eq.A7}
\hat{A}_{oa} \, := \,  \phi_o - \phi_a
\end{equation}
is a constant function on connected components of $o \cap a$; an elementary computation then yields
\begin{equation}
\label{eq.A6}
g_{oa} \ = \  
\lambda_o \, e^{i\hat{A}_{oa}} \lambda_a^{-1} \ = \ 
\lambda_o e^{i\phi_o} \, e^{-i\phi_a} \lambda_a^{-1}
\ \ \ , \ \ \ 
\lambda_o \, := \, e^{ -i \left\{ \phi_o(x_o) + \int_{\gamma(\bar{x}x_o)} A \right\} }
\, ,
\end{equation}
implying that:
(1) $g$ is trivial as a $\bU(1)$-cocycle \cite[\S I.6]{BT};
(2) by rescaling the phases $\lambda_o$, we may take $g_{oa} = e^{i\hat{A}_{oa}}$.
Therefore we may find local charts $\pi_o : L_A |_o \to o \times \bC$ with 
$e^{i\hat{A}_{oa}} = \pi_o \pi_a^{-1}$, and we have the isomorphism
\begin{equation}
\label{eq.A5}
\vartheta : L_A \to M \times \bC 
\ \ \ : \ \ \ 
\vartheta|_o \, := \, e^{-i\phi_o} \pi_o \ , \ \forall o \, .
\end{equation}
%
When $M$ is a spin manifold we set $D_AM := DM \otimes L_A$,
and call $D_AM$ \emph{the Dirac bundle twisted by $A$}.
If there is an isomorphism of the type (\ref{eq.A5}), then with an abuse of notation we write
\begin{equation}
\label{eq.A5a}
\pi_o : D_AM |_o \to DM|_o
\ \ \ , \ \ \  
\vartheta : D_AM \to DM
\, ,
\end{equation}
for the isomorphisms obtained by tensor product with the identity of $DM$.


\section{Dirac fields and background flat potentials}
\label{sec.B}

We now pass to the relativistic scenario and from now on assume that $M$
is a 4-dimensional connected globally hyperbolic spacetime \cite[\S 2]{GLRV01}.
In this case one can define the \emph{future} $J^+(X) \subset M$ and 
the \emph{past} $J^-(X) \subset M$ of a subset $X \subset M$.
We define the \emph{causal complement} $X^\perp := M \setminus cl(J(X))$,
where $cl$ is the closure and $J(X) := J^+(X) \cup J^-(X)$.
For $X,Y \subset M$, we write $X \perp Y$ whenever $Y \subseteq X^\perp$,
with the physical meaning that no event in $X$ can affect an event in $Y$ and \emph{vice versa}.


A distinguished class of subsets of $M$ is given by the \emph{diamonds}
which are, roughly speaking, the domains of dependence of 3--balls lying in Cauchy hypersurfaces of $M$ \cite{GLRV01,BR08}. 
We denote diamonds by lowercase italic letters $a,o,e, \ldots$.
Diamonds fulfill useful properties:
they are open, relatively compact, connected and simply connected,
causally complete ($a = a^{\perp\perp}$) and with connected causal complement $a^\perp$.
A \emph{path} is by definition a finite sequence of pairs of diamonds 
\[
p \, := \, (o_{n0} o_n) * \ldots * ( o_{k+1} o_{k0}) * (o_{k0} o_k) * (o_k o_{k1}) * \ldots * (o_1 o_{11}) \ ,
\]
such that $o_{k1} , o_{k0} = o_{k+1,1} \subseteq o_k$ for all indices $k=1,\ldots,n$;
then we write 
\[
p : e \to a \ \ \ , \ \ \ e := o_{11} \, , \, a := o_{n0} \, .
\]
A path $p$ is said to be an \emph{approximation}
of a curve $\gamma : [0,1] \to M$ whenever $\gamma = \gamma_n * \ldots * \gamma_1$ and,
for all $k=1,\ldots,n$:
(1) the curve $\gamma_k$ has image in $o_k \subset M$;
(2) $\gamma_k(0) \in o_{k1}$ and $\gamma_k(1) \in o_{k0}$.
In this case we write
$p = p(\gamma)$.
Any curve $\gamma$ admits an approximation $p(\gamma)$, and any path $p$ is the approximation of a curve $c(p)$:
if $p$ is of the type $p : a \to a$, then we may take 
$c(p) := \ldots * \ell(o_{k+1} o_{k0}) * \ell(o_{k0} o_k) * \ell(o_k o_{k1}) * \ldots$,
where the loops 
$\ell(o_{ki} o_k) \equiv$ $\ell(o_{ki} o_k;x)$, $x \in o_{ki} = o_{ki} \cap o_k$, $k=1,\ldots,n$, $i=0,1$,
are defined as in (\ref{eq.A2a}). By \cite{Ruz05}, we have the equality of homotopy classes
\begin{equation}
\label{eq.B7}
[\ell] \ = \ 
[c(p(\ell))] \ = \ 
\ldots * [\ell(o_{k+1} o_{k0})] * [\ell(o_{k0} o_k)] * [\ell(o_k o_{k1})] * \ldots \, .
\end{equation}

Now, in our hypothesis for $M$ it is proved that there is (at least) a spin structure \cite{Ish78},
thus we have a Dirac bundle $DM$
%
%
%
and a Dirac operator 
$i\slashed{\nabla} : \mS_c(DM) \to \mS_c(DM)$.
Using a Cauchy hypersurface of $M$ and the fundamental solution of the free Dirac equation 
$\{ i\slashed{\nabla} - m \} f = 0$, 
we get a scalar product $\lb \cdot \, , \, \cdot \rb$ on $\mS_c(DM)$,
and applying a quasi-free state we obtain a free Dirac field \cite{Dim82,Dap09},
\begin{equation}
\label{eq.B3}
\psi : \mS_c(DM) \to \mB(\mH) \ \ , \ \ \{ i\slashed \nabla - m \} \psi \ = \ 0
\, ,
\end{equation}
fulfilling, for all $f,f' \in \mS_c(DM)$, the CARs
\begin{equation}
\label{eq.B4}
\{ \psi(f)^* \, , \, \psi(f') \} \ = \ \lb f , f' \rb \, \bI
\ \ \ , \ \ \ 
\{ \psi(f) \, , \, \psi(f') \} \ = \ \{ \psi(f)^* \, , \, \psi(f')^* \} \ = \ 0 \, .
\end{equation}
Let now $A \in Z_{dR}^1(M)$. We look for a Dirac field $\psi_A$ solving the Dirac equation
with interaction term given by $A$. To this end, we note that for any diamond $o$ we have the local primitive
$\phi_o \in C^\infty(o)$, $d\phi_o = A_o$, thus defining
\begin{equation}
\label{eq.B5}
\psi_o : \mS_o(DM) \to \mB(\mH) \ \ \ , \ \ \ \psi_o(f) \, := \, \psi(e^{-i\phi_o}f)
\end{equation}
\cite[\S 4.2]{VasQFT}, and applying (\ref{eq.B3}), for all $f \in \mS_o(DM)$ we find
\begin{equation}
\label{eq.B6}
\begin{array}{lcl}
0 & = &
\psi( (i\slashed{\nabla} - m) (e^{-i\phi_o} f )) \ = \\ & = &
\psi( \slashed{A}(e^{-i\phi_o} f ) + ie^{-i\phi_o}\slashed{\nabla}f - me^{-i\phi_o}f ) \ = \\ & = &
\psi_o( (\slashed{A} + i\slashed{\nabla} - m) f ) \, .
\end{array} 
\end{equation}
We conclude that $\psi_o$ solves the interacting Dirac equation for spinors $f$ supported in $o$.
Now, by (\ref{eq.A7}) we have
\begin{equation}
\label{eq.B1}
\psi_{o'} \ = \ e^{-i\hat{A}_{o'o}} \psi_o \ \ \ , \ \ \ o \subseteq o' \, ,
\end{equation}
thus we find an obstruction to glue the fields $\psi_o$ and get the desired global solution $\psi_A$.
To solve the problem, we consider the twisted Dirac bundle $D_AM$ \S \ref{sec.A} and note that,
since $L_A$ has transition maps $e^{i\hat{A}_{o'o}}$, for any $o$ we have the "local charts"
$\pi_o : D_AM |_o \to DM|_o$
such that $\pi_{o'} \pi_o^{-1} f = e^{i\hat{A}_{o'o}}f$, $f \in \mS_o(DM)$.
If $\varsigma \in \mS_o(D_AM)$, then defining 
$\varsigma_o := \pi_o \varsigma$, $\varsigma_o \in \mS_o(DM)$,
we find that
\begin{equation}
\label{eq.B8}
\psi_o(\varsigma_o) \ = \ \psi(e^{-i\phi_o}\pi_o(\varsigma)) \, \stackrel{ (\ref{eq.A5}) }{=} \, \psi(\vartheta(\varsigma)) \, ,
\end{equation}
is independent of $o$. Thus defining $\psi_A(\varsigma) := \psi(\vartheta(\varsigma))$ we obtain 
a Dirac field
{\footnote{
$\psi_A$ shares with twisted fields in the sense of \cite{Ish78a} the property of being defined on a twisted bundle.
}}
\begin{equation}
\label{eq.B2}
\psi_A : \mS_c(D_AM) \to \mB(\mH)
\ \ , \ \ 
\{ i\slashed \nabla + \slashed A - m \} \psi_A \ = \ 0 \, .
\end{equation}
On the converse, suppose that a Dirac field (\ref{eq.B2}) is given.
Then computations analogous to (\ref{eq.B6}) and (\ref{eq.B8}) show that defining 
$\psi_o(f) := \psi_A(\pi_o^{-1}f)$, $f \in \mS_o(DM)$,
we get a family of fields such that:
(1) the relations (\ref{eq.B1}) hold;
(2) the field $\psi(f) := \psi_o(e^{i\phi_o}f)$, $f \in \mS_o(DM)$, fulfils the free Dirac equation and is independent of $o$.
In conclusion, we have proved:
\begin{thm}
\label{thm.B1}
Let $A \in Z_{dR}^1(M)$ and $D_AM$ denote the associated
twisted Dirac bundle. Then there is a field $\psi_A$ fulfilling (\ref{eq.B2}),
and this is equivalent to saying that there is a family of fields
$\psi_o : \mS_o(DM) \to \mB(\mH)$,
defined for any diamond $o \subset M$,
fulfilling (\ref{eq.B1}) and gauge-equivalent to a free Dirac field $\psi$.
\end{thm}


\section{Topological sectors}
\label{sec.C}

In the algebraic approach to quantum field theory the physics of a quantum system
is encoded by the family of Von Neumann algebras generated by observables localized in spacetime regions,
typically diamonds \cite{Haa}. 
As an intermediate step, we illustrate nets of the larger algebras generated
by not necessarily observable field operators.

\medskip

Let $\mH$ be a Hilbert space and $G \subset \mU(\mH)$ a group compact and metrizable
under the strong topology, such that there is $\delta \in G \cap G'$, $\delta^2 = \bI$
(here, for any subset $\mS$ of $\mB(\mH)$, $\mS'$ denotes the commutant).
Then $\delta$ splits $\mB(\mH)$ into \emph{bosonic} operators
$B \in \mB(\mH)$ such that $\delta B \delta = B$,
and \emph{fermionic} operators
$F \in \mB(\mH)$ such that $\delta F \delta = -F$.
A \emph{field net $\mF$ with gauge group $G$} is given by a collection of Von Neumann algebras
$\mF_o \subset \mB(\mH)$ assigned in correspondence with diamonds $o \subset M$, such that:
\begin{itemize}
\item $T \in \mF_o'$ for all $o$ implies $T = z\bI$ for some $z \in \bC$ (irreducibility);
\item $\mF_o \subseteq \mF_{o'}$ for $o \subseteq o'$ (isotony);
\item $\alpha_g(T) := gTg^* \in \mF_o$ for all $T \in \mF_o$ and for all diamonds $o$ (gauge action);
\item $[B,B'] = [B,F'] = [F,B'] = \{ F,F' \} = 0$, for all 
      operators $B$ (bosonic), $F$ (fermionic) in $\mF_o$,
      and $B'$ (bosonic), $F'$ (fermionic) in $\mF_e$, with $e \perp o$
      (normal commutation relations);
\item For any diamond $o$ and projection $E \in \mF_o$, there is an isometry $V \in \mF_o$ such that
      $V^*V = \bI$, $VV^* = E$ (Type ${\mathrm{III}}_1$ "Borchers property").
\end{itemize}
The most common way to construct a field net is to start from a Wightman field
$\phi$
and to define the Von Neumann algebras as the double commutants
{\footnote{For simplicity, here we bypass the technical complications that arise for $\phi(f)$ unbounded.}}
\[
\mF_o \, := \, \{ \phi(f) \, , \, \phi(f)^* \ : \ \supp(f) \subseteq o \}'' \, ,
\]
thus the former four properties are natural by keeping in mind Wightman's approach.
The type ${\mathrm{III}}_1$ property is a manifestation of positivity of energy and locality, 
and has been verified both in Minkowski spacetime (for the basic physical requirements \cite{BdAF87})
and in curved spacetimes (for the nets of the free Klein-Gordon and Dirac fields \cite{Ver97,dAH06}).

\medskip

We now consider the Hilbert space $\mH^0$ of $G$-invariant vectors and note that any 
$T \in \mF_o \cap G'$ restricts to the operator 
$\pi^0(T) := T | \mH^0 \in \mB(\mH^0)$.
We make the standard assumption that the so-obtained mapping $\pi^0$ is faithful
{\footnote{This property is trivially verified in the case of the free Dirac field in which we are interested.}}
and, to be concise, we write
$t := \pi^0(T)$, $\forall T \in \mF_o \cap G'$.
Thus we define the \emph{observable net} $\mR$, 
\[
\mR_o \, := \, \{ t \, : \, T \in \mF_o \cap G' \} \, \subseteq \mB(\mH^0) \ \ \ , \ \ \ \forall o \, ,
\]
for which irreducibility (in $\mB(\mH^0)$), isotony and ${\mathrm{III}}_1$ property are required.
The normal commutation relations imply $\mR_o \subseteq \mR_e'$, $o \perp e$,
and these relations can be strengthened by requiring \emph{Haag duality}
\begin{equation}
\label{eq.C8}
\mR_o \, = \, \cap_{e \perp o} \mR_e' \ .
\end{equation}
In the case of the free Dirac field (\ref{eq.B3})
we take $G \simeq \bU(1)$ given by the second quantization map $U_\zeta \in \mU(\mH)$, $\zeta \in \bU(1)$,
set $\delta = U_{-1}$, and define, for any diamond $o$,
\begin{equation}
\label{eq.C12}
\mF_o \, := \, \{ \psi(f) \, , \, \psi(f)^*  \ : \ \supp(f) \subseteq o \}'' \, ,
\end{equation}
with gauge action $\alpha_\zeta(T) := U_\zeta T U_\zeta^*$, $T \in \mF_o$.
Irreducibility, isotony and the ${\mathrm{III}}_1$ property hold for the associated observable net $\mR$ \cite{dAH06}.
Haag duality, notwithstanding it is widely believed to be true for $\mR$,
will not be used in the present paper.

\paragraph{Twisted field nets.}
Let $\mF$ be a field net with gauge group $G$. A \emph{twist} of $\mF$ induced by a morphism
$\sigma : \pi_1(M) \to G$
is given by the family of *-monomorphisms
\begin{equation}
\label{eq.C13}
\sigma_{o'o} \, := \, \alpha(g_{o'o}) : \mF_o \to \mF_{o'} 
\ \ \ , \ \ \ 
o \subseteq o'
\, ,
\end{equation}
where $g_{o'o} \in G$ are defined by (\ref{eq.A2})
{\footnote{Recall that $o = o \cap o'$ is simply connected so $g_{o'o}$ is a constant $G$-valued map.}},
\cite{VasQFT}.
%
%
We denote the twisted net defined by (\ref{eq.C13}) by $\mF^\sigma$.
At the mathematical level, we note that $\mF^\sigma$ is a precosheaf.
This means that, to perform algebraic operations on operators of different algebras,
one must take account of (\ref{eq.C13}): for example, the product of operators
$T \in \mF_o$, $T' \in \mF_{o'}$, $o \subseteq o'$, is given by
\begin{equation}
\label{eq.C14}
\sigma_{o'o}(T) T' \, .
\end{equation}
By the results in \cite{BR08,RV11,RV14a,RV14b}, we may regard $\mF^\sigma$ as a representation of $\mF$ over the flat Hilbert
bundle $H \to M$ with fibre $\mH$ and monodromy $\sigma(\ell) \in G \subset \mU(\mH)$, $[\ell] \in \pi_1(M)$.

To understand why on earth one should twist a field net, we consider the case where $\mF$ is the net 
of the free Dirac field (\ref{eq.C12}). We note that by Theorem \ref{thm.B1} the interacting field (\ref{eq.B2})
is characterized by the family $\{ \psi_o \}$, which obviously generates the same net as $\psi$ since 
$\psi_o(f) = \psi(e^{-i\phi_o}f)$
for $\supp(f) \subseteq o$ and for all $o$.
Yet there is an additional information carried by the fields $\psi_o$ that does not appear in the Von Neumann algebras.
This is given by (\ref{eq.B1}), where the shift phase $e^{-i\hat{A}_{o'o}}$ appears when one compares $\psi_o$ with $\psi_{o'}$.
Extending this relation to the Von Neumann algebras generated by the $\psi_o$'s, we get the *-monomorphisms
\begin{equation}
\label{eq.C15}
\alpha(e^{-i\hat{A}_{o'o}}) : \mF_o \to \mF_{o'} \ \ \ , \ \ \ o \subseteq o' \, .
\end{equation}
Thus (the phases $\lambda_o$ being irrelevant) by (\ref{eq.A6}) we may regard (\ref{eq.C15}) as the twist induced by the morphism
$\sigma(\ell) := \exp - \oint_\ell A$,
and $\psi_A$ defines the twisted field net $\mF^\sigma$.
On the converse, if one starts from a twist $\mF^\sigma$ such that the gauge potential $A$ (\ref{eq.A4ab}) fulfils $dA=0$,
then by (\ref{eq.A7}) we have local primitives $\phi_o$ which can be used to define the family (\ref{eq.B5}) 
and therefore the field $\psi_A$.
In conclusion:
\begin{prop}
\label{prop.C1}
Let $\mF$ denote the net of the free Dirac field.
Then, any interacting Dirac field $\psi_A$, $dA=0$, defines a twist $\mF^\sigma$ with $\sigma(\ell) := \exp - \oint_\ell A$.
On the converse, any twist $\mF^\sigma$ such that the flat gauge potential $A$ (\ref{eq.A4ab}) fulfils $dA=0$
defines an interacting Dirac field $\psi_A$.
\end{prop}

The previous result illustrates the physical meaning of the operation of twisting a field net.
Motivated by that, and recalling (\ref{eq.A4ab}), we interpret $\sigma$ as an Aharonov-Bohm phase,
even for $dA \neq 0$ and, more in general, for $G$ non-Abelian.
In the next paragraph and in \S \ref{sec.D} we shall show how $\sigma$ is encoded in the observable net $\mR$.

\paragraph{Superselection sectors.} The belief of algebraic quantum field theory is that
the physics of a quantum system is encoded by $\mR$ (rather than $\mF$). This turns out to be true in 
Minkowski spacetime, where one is able to reconstruct $\mF$ starting from $\mR$ \cite{DR90}.
The information needed for the reconstruction of $\mF$ is given by the \emph{superselection structure},
that is the family of Hilbert space representations of $\mR$ that are physically relevant
according to a criterion that shall be explained below.

At the present time no analogue of the result \cite{DR90} is known in curved spacetimes and,
indeed, this question is a motivation for the present paper.
That we should be able to reconstruct the field net $\mF$ is expectable by the ordinary theory of
superselection sectors; the question is whether $\mR$ contains the information needed to reconstruct
twisted field nets, which encode interactions with background flat potentials. 
We are going to show that the answer is affirmative for the free Dirac field,
and the needed information is given by the superselection structure of $\mR$ in the sense of \cite{BR08}
that we illustrate in the following lines.

Superselection sectors in Minkowski spacetime are defined by suitable representations 
of the $C^*$-algebra $C^*\mR$ generated by $\cup_o \mR_o$.
In curved spacetimes, one defines for all diamonds $o$ "local sectors" 
\begin{equation}
\label{eq.C9}
\pi_o : \mR_o \to \mB(\mH^0)
\end{equation}
and then requires that they are unitarily equivalent
each other by means of a family of unitary operators
$z = \{ z_{o'o} \in \mU(\mR_{o'}) \}_{o \subseteq o'}$
\cite{GLRV01,Ruz05,BR08}.
In precise terms, we define sectors as pairs $(z,\pi)$, where $z$ fulfils
\begin{equation}
\label{eq.C1}
z_{o''o'} z_{o'o} \, = \, z_{o''o} \ \ \ , \ \ \ \forall o \subseteq o' \subseteq o'' \, .
\end{equation}
For convenience we set $z_{oo'} := z_{o'o}^*$ for $o \subseteq o'$,
and note that (\ref{eq.C1}) implies $z_{oo} \equiv \bI$. 
Given a path 
$p : a \to o$, $p = (oo_n) * \ldots * (o_1a)$,
we define the unitary operator
\begin{equation}
\label{eq.C3}
z_p \, := \, z_{oo_n} \cdots z_{o_1 a} \ .
\end{equation}
The symbol $\pi$ stands for the family of representations (\ref{eq.C9}), which fulfil
\begin{equation}
\label{eq.C2}
\pi_{o'}(t) \ = \ z_{o'o} \, \pi_o(t) \, z_{oo'}
\ \ \ , \ \ \ 
\forall o \subseteq o' \, , \, t \in \mR_o \subseteq \mR_{o'}
\, .
\end{equation}
We call $\pi$ the \emph{charge} defined by $(z,\pi)$.
Instead, $z$ is called the \emph{charge transporter} and, in mathematical terms,
it is interpreted as a parallel transport for charges localized in diamonds $o$.
Sectors are required to fulfil the following properties:
\begin{enumerate}
\item \emph{Net structure preservation}: 
      $\pi_o(\mR_o) \subseteq \mR_o$ for all $o$.  
      This allows to perform compositions of charges, defined by $\pi'_o \circ \pi_o$;
\item \emph{Localization}: 
      given a diamond $o$, for any $a \supset o$ we extend $\pi_o$ by defining the *-morphism
      $\pi^a_o : \mR_a \to \mR_a$, $\pi^a_o(t) := z_{oa} \pi_a(t) z_{ao}$, $\forall t \in \mR_a$.
      Note in fact that by (\ref{eq.C2}) we have $\pi^a_o(t) = \pi_o(t)$ for all $t \in \mR_o$.
      We then require that
      \[
      \pi^a_o(t') = t' \ \ \ , \ \ \ \forall t' \in \mR_e \subseteq \mR_a \ , \ e \subset a \ , \ e \perp o \ .
      \]
      This property is interpreted as the fact that $\pi_o$ is an excitation localized in $o$
      of the reference representation $\pi^0$, in the sense that it leaves invariant observables
      localized in the causal complement of $o$.
      This fits the classical notion of sector in the sense of Doplicher, Haag, Roberts \cite{Haa,DR90}.
\end{enumerate}
The above definition of sector is redundant when $\mR$ fulfils Haag duality. 
In fact, in this case we can recover 
\begin{equation}
\label{eq.C4}
\pi_o(t) \ = \ z_p \, t \, z_p^* \ \ \ , \ \ \ t \in \mR_o \, ,
\end{equation}
where $p : e \to o$ is an arbitrary path with $e \perp o$.
Using (\ref{eq.C4}) and Haag duality, it is obviously verified that (\ref{eq.C2}),
net structure preservation and localization are fulfilled,
thus the family $z$ is sufficient to define sectors, as done in \cite{BR08}.

\medskip

Let now $(z,\pi)$ be a sector. Then the mapping
\begin{equation}
\label{eq.C10}
\pi_1(M) \to \mU(\mH^0) \ \ \ , \ \ \ [\ell] \mapsto z_\ell := z_{p(\ell)} \, ,
\end{equation}
is well-defined as $z_\ell = z_{\ell'}$ for $\ell$ homotopic to $\ell'$,
and defines a unitary representation of $\pi_1(M)$;
moreover, there are $n \in \bN$, a unitary operator $U : \mH^0 \to \mH^0 \otimes \bC^n$
and a representation 
\begin{equation}
\label{eq.C11}
\sigma^z : \pi_1(M) \to \bU(n)
\end{equation}
such that
$z_\ell =$ $U^* \, ( \bI \otimes \sigma^z(\ell) ) \, U$, $\forall [\ell] \in \pi_1(M)$ \cite{BR08}.
The invariant $n$ is called the \emph{topological dimension} of $(z,\pi)$, and $\sigma^z$ the 
\emph{topological component}. We shall argue in \S \ref{sec.D} that $\sigma^z$ should take values
in the image of a representation of the gauge group $G$.
%
%
A sector is said to be \emph{topological} whenever (\ref{eq.C11}) is not trivial,
and \emph{topologically trivial or DHR-sector} otherwise.
DHR-sectors correspond to representations of Fredenhagen's universal algebra $\vec{\mR}$ \cite{Fre90},
whilst topological sectors do not.
Indeed, topological sectors correspond to representations of $\mR$ on flat Hilbert bundles $H^0 \to M$ with fibre $\mH^0$
and monodromy (\ref{eq.C10}) \cite{BR08,RV11,RV14a,RV14b}.
By (\ref{eq.C11}) and (\ref{eq.A4}), any $(z,\pi)$ defines a flat gauge potential
$A^z \in \mS(T^*M \otimes \efu(n)_{\sigma^z})$ whose path-ordered integral is the topological component of $(z,\pi)$.
Examples of topological sectors for massive boson fields have been constructed in low dimension \cite{BFM09}
In the following we consider the field net of the free Dirac field in four spacetime dimensions and its sectors with topological dimension $1$, 
postponing comments on the general case in \S \ref{sec.D}.

\paragraph{Free Dirac field.}
We now consider the net (\ref{eq.C12}) having gauge group $G \simeq \bU(1)$.
By construction, $\mF$ splits into the spectral subspaces
\[
\mF_o^\kappa \, := \, \{ T \in \mF_o \, : \, \alpha_\zeta(T) = \zeta^\kappa T \, , \, \forall \zeta \in \bU(1) \} 
\ \ \ , \ \ \ 
\kappa \in \bZ \, ,
\]
with $\mF_o^0 = \mF_o \cap G'$.
Clearly $\psi(f) \in \mF_o^1$ for $\supp(f) \subseteq o$ and $\kappa$ is interpreted as the electric charge.
In the following lines we construct a set of DHR-sectors of $\mR$ labelled by $\kappa$.

Let $f \in \mS_o(DM)$ such that $\lb f , f \rb = 1$. 
Then writing the CARs (\ref{eq.B4}) for $f=f'$ we find that $\psi(f)^*\psi(f)$, $\psi(f)\psi(f)^* \in \mF_o^0$ are projections,
and by the ${\mathrm{III}}_1$ property there are $V_o , W_o \in \mF_o^0$ such that 
$\varphi_o := W_o^* \psi(f) V_o \in \mF_o^1$
is unitary. The normal commutation relations then imply
\begin{equation}
\label{eq.C7}
\varphi_o^* \varphi_e \, = \, - \varphi_e \varphi_o^*
\ \ \ , \ \ \ 
\varphi_o \varphi_e \, = \, \varphi_e \varphi_o
\ \ \ , \ \ \ 
\varphi_o T \, = \, T \varphi_o \, ,
\end{equation}
for all $o \perp e$ and $T \in \mF_e^0$.
We define the unitary operators
\begin{equation}
\label{eq.C5}
z^1_{o'o} \, := \, \pi^0(\varphi_{o'} \varphi_o^*) \ \ \ , \ \ \ o \subseteq o' \, .
\end{equation}
It is clear that $z^1_{o'o}$ belongs to $\mU(\mR_{o'})$ and fulfils (\ref{eq.C1});
if $p : e \to o$ is a path, then by (\ref{eq.C3}) and (\ref{eq.C5}) we have $z^1_p = \pi^0(\varphi_o\varphi_e^*)$.
Defining the charge
\begin{equation}
\label{eq.C6}
\pi^1_o(t) \, := \, \pi^0(\varphi_o T \varphi_o^*) \ \ \ , \ \ \ \forall o \, , \, t \in \mR_o \, ,
\end{equation}
and using (\ref{eq.C7}),
we find that (\ref{eq.C2}) and the properties of net structure preservation and localization are obviously fulfilled.
Thus $(z^1,\pi^1)$ is a sector.
Note that, even if we have not checked Haag duality, 
(\ref{eq.C4}) holds,
\[
z^1_p \, t \, (z^1_p)^* \ = \ 
\pi^0(\varphi_o \varphi_e^*) \, t \, \pi^0(\varphi_e \varphi_o^*) \ = \ 
\pi^0(\varphi_o T \varphi_o^*) \ = \ \pi^1_o(t)
\ \ \ , \ \ \ 
t \in \mR_o  \, , \, e \perp o \, ,
\]
having used again (\ref{eq.C7}). 
Thus $(z^1,\pi^1)$ is a DHR-sector carrying charge $1$;
since $\pi^1$ is implemented by the family of unitaries $\{ \varphi_o \}$ (and not by a family of multiplets of isometries as in \S \ref{sec.D}),
we can conclude that $(z^1,\pi^1)$ has topological dimension $1$.
Analogously, using powers $\psi(f)^n$, $(\psi(f)^*)^n$, $n \in \bN$, in place of $\psi(f)$,
one constructs sectors $(z^\kappa,\pi^\kappa)$, $\kappa \in \bZ$,
carrying charge $\kappa$ and having topological dimension $1$.

Now, we are interested in generic sectors $(z,\pi)$ having topological dimension $1$ and such that $\pi$ is equivalent
to the "fundamental charge" $\pi^1$, in the sense that for all $o$ there is a unitary $u_o \in \mU(\mR_o)$ with 
$\pi_o(t) = u_o \pi^1_o(t) u_o^*$, $t \in \mR_o$.
We denote the set of such sectors by ${\bf sect}^1(\mR)$.
\begin{thm}
\label{thm.C1}
Let $\mF$ be the field net of the free Dirac field and $\mR$ denote the observable net.
Then any $(z,\pi) \in {\bf sect}^1(\mR)$ defines a flat gauge potential $A^z \in \mS(T^*M)$ and a twisted field net $\mF^{\sigma^z}$.
On the converse, any twist $\mF^\sigma$ with flat gauge potential $A \in \mS(T^*M)$ defines  a sector $(z,\pi) \in {\bf sect}^1(\mR)$.
In particular, sectors $(z,\pi) \in {\bf sect}^1(\mR)$ such that $dA^z = 0$ 
are in one-to-one correspondence with interacting fields $\psi_A$, $dA = 0$.
\end{thm}

\begin{proof}
Let $(z,\pi) \in {\bf sect}^1(\mR)$. We already know that there is a topological component
$\sigma^z : \pi_1(M) \to \bU(1)$
and a flat gauge potential $A^z \in \mS(T^*M)$, thus we can proceed to the construction of the twist $\mF^{\sigma^z}$.
On the converse, given a twist $\mF^\sigma$, $\sigma : \pi_1(M) \to \bU(1)$, with flat gauge potential $A \in \mS(T^*M)$,
we define the pair $(z^\sigma,\pi^1)$ where, for all $o \subseteq o'$, 
\begin{equation}
\label{eq.C16}
z^\sigma_{o'o} \ := \ 
\pi^0(\varphi_{o'} \sigma_{o'o}(\varphi_o)^*) \ = \ 
g_{o'o} \, z^1_{o'o} \, ,
\end{equation}
see (\ref{eq.A2}), (\ref{eq.C13}), (\ref{eq.C5}) and (\ref{eq.C6}). 
A straightforward check shows that $(z^\sigma,\pi^1) \in {\bf sect}^1(\mR)$.
Finally, the last part of the theorem follows by Proposition \ref{prop.C1}.
\end{proof}

The constructions of the previous theorem are the inverse one of each other up to equivalence in the sense of \cite[\S 2]{BR08};
in particular, if $(z,\pi) \in {\bf sect}^1(\mR)$ then the sector $(z^\sigma,\pi^1)$ defined for $\sigma = \sigma^z$
is equivalent to $(z,\pi)$, as can be proved by using the notion of splitting \cite[\S 6.1]{BR08}.
Using (\ref{eq.B7}) and (\ref{eq.C16}), we deduce the topological component of $(z^\sigma,\pi^1)$:
\begin{equation}
\label{eq.C17}
z^\sigma_\ell \ = \ 
\cdots g_{o_{k+1} o_{k0}} \cdot g_{o_k o_{k0}}^{-1} \cdot g_{o_k o_{k1}} \cdots \, z^1_{aa} \ = \ 
\sigma(\ell) \bI \ = \
\exp \oint_\ell A \cdot \bI \, .
\end{equation}
To appreciate the role of the phase operators $z^\sigma_\ell$ in charged states, it is convenient to come back
to the physical Hilbert space $\mH$ and consider the operators
$\hat{z}_{o'o} :=$ $g_{o'o} \varphi_{o'} \varphi_o^* \in \mU(\mF_o^0)$, $o \subseteq o'$,
such that $\pi^0(\hat{z}_{o'o}) = z_{o'o}$. 
Given the reference vector $\omega \in \mH$
we define $v_o := \varphi_o \omega$, obtaining vectors with charge 1
interpretable as excitations of $\omega$ localized in $o$. Applying the operators
$\hat{z}_p$, $p : a \to o$,
defined as in (\ref{eq.C3}), we find 
\[
\hat{z}_p \, v_a \ = \ g_p \cdot \varphi_o \varphi_a^* \, \varphi_a \omega \ = \ g_p \cdot v_o
\ \ \ , \ \ \ 
g_p \in \bU(1)
\, .
\]
Thus $\hat{z}_p$ transports the state defined by $v_a$ into the one defined by $v_o$.
The transport is path-dependent, in fact given $q : a \to o$ we have the transition amplitude
\begin{equation}
\label{eq.C18}
\lb \hat{z}_q v_a \, , \, \hat{z}_p v_a \rb \ = \ 
\lb v_a \, , \, \hat{z}_{\ovl{q} \, *p} \, v_a \rb \ = \ 
\exp \oint_\ell A \, ,
\end{equation}
where $\ovl{q} : o \to a$ is the opposite of $q$, $\ovl{q}*p$ is the path composition \cite{RV11},
and the loop $\ell$ is defined such that $p(\ell) = \ovl{q}*p$.


\section{Non-abelian phases}
\label{sec.D}

Non-Abelian Aharonov-Bohm phases have been considered by many authors (see \cite{Hov86} and related references) and,
as we saw in the previous section, naturally appear in the context of sectors
in the guise of the topological component (\ref{eq.C11}).
A complete discussion of non-Abelian phases is beyond the scope of the present paper, 
nevertheless it is instructive to give a sketch of the way they give rise to superselection sectors.

Our starting point is a field net $\mF$ with compact gauge group $G$. To fix ideas, we may assume that $\mF$ is generated by a
generalized Dirac field
\[
\psi : \mS_c(DM \otimes \bC^d) \to \mB(\mH) \, ,
\]
with $G \subseteq \bU(d)$ appearing in an irreducible fundamental representation. 
Using the analysis of spectral subspaces \cite{DR90} and the ${\mathrm{III}}_1$ property, 
one can prove that for any finite-dimensional unitary representation $\rho : G \to \bU(n)$ (concisely: $\rho \in \hat{G}$),
and for any diamond $o$, there is a multiplet of isometries $\{ \varphi_{o,i} \in \mF_o \}_{i=1,\ldots,n}$ such that
\begin{equation}
\label{eq.E1}
\varphi_{o,i}^* \varphi_{o,j} \, = \, \delta_{ij} \, \bI
\ \ , \ \ 
\sum_i^n \varphi_{o,i} \varphi_{o,i}^* = \bI
\ \ , \ \ 
\alpha_g(\varphi_{o,i}) \, = \, \sum_j^n \rho(g)_{ij} \, \varphi_{o,j}
\, .
\end{equation}
%
%
%
%
Therefore, at varying of the diamond $o$, we define the charge
\begin{equation}
\label{eq.E2}
\pi^\rho_o(t) \, := \, \pi^0 \left( \sum_i^n \varphi_{o,i} T \varphi_{o,i}^* \right)
\ \ \ , \ \ \ 
t \in \mR_o
\, .
\end{equation}
To construct a charge transporter for $\pi^\rho$, we consider
a twist $\sigma : \pi_1(M) \to G$ and the associated twisted field net $\mF^\sigma$.
Given our $\rho \in \hat{G}$, we have that $\rho(G)$ is a compact Lie group with Lie algebra $\efg_\rho$,
and performing the composition $\rho \circ \sigma$ defines the principal $\rho(G)$-bundle $P_{\rho\sigma}$ and
the $\efg_\rho$-bundle $\efg_{\rho\sigma} := P_{\rho\sigma} \times_\ad \efg_\rho$.
By the results in \S \ref{sec.A}, we have the flat gauge potential $A^{\rho\sigma} \in \mS(T^*M \otimes \efg_{\rho\sigma})$ such that
\begin{equation}
\label{eq.E3}
\rho(\sigma(\ell)) \ = \ \xP\exp \oint_\ell A^{\rho\sigma} \, \in \rho(G) \subset \bU(n)
\ \ \ , \ \ \ 
\forall [\ell] \in \pi_1(M) \, .
\end{equation}
On these grounds, generalizing (\ref{eq.C16}) we define the unitaries
\begin{equation}
\label{eq.E4}
z^{\rho\sigma}_{o'o} \ := \
\pi^0 \left( \sum_i^n  \varphi_{o',i} \, \sigma_{o'o}(\varphi_{o,i})^* \right) 
\ \ \ , \ \ \ 
o \subseteq o'
\, .
\end{equation}
Using (\ref{eq.E1}) it can be verified that $(z^{\rho\sigma},\pi^\rho)$ is a sector with 
topological component (\ref{eq.E3}) (the calculation is analogous to (\ref{eq.C17})).
In this way we have a family of sectors of $\mR$ labeled by $\rho$ (alike classical DHR-sectors)
and $\sigma$, this last encoding the "Wilson loops" (\ref{eq.E3}).
Thus we assembled our sectors in such a way that the topological component takes values 
in a representation of the gauge group $G$, and our working hypothesis is that this is true for any sector $(z,\pi)$:
to advocate this point, we recall Theorem \ref{thm.C1} and the subsequent computation (\ref{eq.C17}),
where it is shown the equivalence between flat gauge potentials and topological components for $G \simeq \bU(1)$.
The proof that any $(z,\pi)$ is of this type would require a reconstruction theorem of the type \cite{DR90}
which is, at the present time, lacking.


\section{Conclusions}
\label{sec.E}

It is an established result that DHR superselection sectors of an observable net label quantum charges
corresponding to representations of the global gauge group \cite{DR90}. 
For example, in the present paper we exhibited a family of DHR sectors of the observable net of the free Dirac field,
labelling the electric charge.

When the spacetime has a non-trivial fundamental group a further class of superselection
sectors is available (topological sectors, \cite{BR08}), and we proved that in the case
of the free Dirac field these sectors are interpreted in
terms of interactions with background flat electromagnetic potentials.
The corresponding topological components (\ref{eq.C11}) are then interpreted as Aharonov-Bohm phases
affecting the parallel transport of charges, and yield the expected value (\ref{eq.C17}),( \ref{eq.C18}).
Note that the potential is reconstructed by the phase, \S \ref{sec.A};
this is satisfactory from the physical point of view,
because it is the phase that actually appears in experiments.

These results imply that the observable net encodes its own background flat interactions.
For example, with the notation of \S \ref{sec.C}, the DHR sector $(z^1,\pi^1)$ encodes the "charge of the electron",
whilst $(z^\sigma,\pi^1)$ is associated to the electron interacting with the potential $A$ defined by $\sigma$.
Thus the presence of the flat background potential is detected at the level of the observable algebra
in terms of the non-trivial parallel transport operators $z^\sigma_\ell$,
whilst for the DHR-sector we have $z^1_\ell \equiv \bI$.

\medskip

We remark that the Aharonov-Bohm effect has been discussed as an aspect of the
quantum electromagnetic field in the setting of locally covariant quantum field theory \cite{SDH15}.
There, a peculiar property of the family of algebras assigned to spacetime manifolds is that the "inclusion" morphisms
assigned to spacetime embeddings may not be injective,
as a consequence of the fact that Aharonov-Bohm phases may disappear in bigger, \emph{simply connected} spacetimes.
In the approach of the present paper this eventuality is avoided, because we have a fixed background spacetime $M$,
so that loop observables of the type $z_\ell$ are non-trivial whenever the homotopy class of $\ell$ is non-trivial in the topology of $M$.
Moreover, diamonds are simply connected, thus no topological effect related to the fundamental group may appear in the local algebras $\mR_o$
{\footnote{
Anyway, this point deserves a deeper discussion, in particular on the issue of locality,
which is postponed to the above-mentioned work in progress by Dappiaggi, Ruzzi and the author.}}.

\medskip

In a non-relativistic setting,
we note that our results qualitatively fit those found by Morchio and Strocchi for finitely many degrees of freedom \cite{MS07}:
on a generic manifold $M$, uniqueness of the representation of the Heisenberg relations breaks down, 
and representations of the Weyl algebra are labelled by Hilbert space representations $\sigma$ of $\pi_1(M)$.
Yet in the above-cited reference no description of $\sigma$ is given in terms of flat background potentials,
and this may be an interesting point to investigate
{\footnote{
As a matter of fact, such an interpretation would fit results proven in the setting of quantization 
\emph{via} path integrals \cite{Hov80}.
The author would like to thank P.A. Horv\'athy for this reference.}}.


\medskip

A point that has not explicitly discussed in this paper is the eventuality that $\pi_1(M)$ is non-Abelian.
This is not an exotic scenario since, for example, adding a second shielded (parallel) solenoid to the Aharonov-Bohm apparatus
yields $\pi_1(N) \simeq \bF_2$, the free group with two generators.
$\pi_1(M)$ being non-Abelian, representations $\sigma : \pi_1(M) \to \bU(n)$, $n > 1$, 
appear in such a way that $\sigma$ takes values in a possibly non-Abelian subgroup $G \subset \bU(n)$.
Thus we fall in the case sketched in \S \ref{sec.D}, that shall be discussed in future publications.

\medskip

Finally we mention that non-flat background potentials can be codified as well in the observable net of the free Dirac field.
In this case we lose homotopic invariance of the phase $\sigma$ and, instead of charge transporters fulfilling (\ref{eq.C1}),
we get more in general connections as in \cite{CRV1,CRV2}.
This scenario may be analysed already in Minkowski spacetime, by considering the net of the free Dirac field
and codifying the interaction with a potential $A \in \mS(T^*\bR^4)$, $dA \neq 0$, by means of a connection $z$
carrying the holonomy $\sigma^z(\ell) = \exp \oint_\ell A$.


\paragraph{Acknowledgement.}
The author is supported in part by OPAL "Consolidate the Foundations".
The author would also like to thank L. Giorgetti for useful remarks.


\end{document}